\documentclass{llncs} 
\usepackage{datetime}
\usepackage{setspace}
\usepackage{amsmath}
\usepackage{amssymb}
\usepackage{graphicx}
\usepackage[nocompress]{cite}
\usepackage{color}
\newtheorem{thm}{Theorem}[subsection]
\newtheorem{cor}[thm]{Corollary}

\newtheorem{prop}[thm]{Proposition}

\newtheorem{defn}[thm]{Definition}

\newtheorem{rem}[thm]{Remark}
\newtheorem{res}[thm]{Result}

\newcommand{\ra}{\rightarrow}
\newcommand{\pr}{\prime}
\newcommand{\dpr}{{\prime\prime}}

\newcommand{\ov}{\overline}

\newcommand{\mc}{\mathcal}
\newcommand{\tbf}{\textbf}
\newcommand{\lel}{\left|\left}
\newcommand{\rer}{\right|\right}

\title{Parametrization of completeness in symbolic
  abstraction of bounded input linear systems}
\author{Santosh Arvind Adimoolam}\institute{}
\begin{document}
\maketitle
\begin{abstract}
  A good state-time quantized symbolic abstraction of an
  already input quantized control system would satisfy three
  conditions: proximity, soundness and completeness.  Extant
  approaches for symbolic abstraction of unstable systems
  limit to satisfying proximity and soundness but not
  completeness.  Instability of systems is an impediment to
  constructing fully complete state-time quantized symbolic
  models for bounded and quantized input unstable systems,
  even using supervisory feedback.  Therefore, in this paper
  we come up with a way of parametrization of completeness
  of the symbolic model through the quintessential notion of
  ``Trimmed-Input Approximate Bisimulation'' which is
  introduced in the paper.  The amount of completeness is
  specified by a parameter called ``trimming'' of the set of
  input trajectories.  We subsequently discuss a procedure
  of constructing state-time quantized symbolic models which
  are near-complete in addition to being sound and proximate
  with respect to the time quantized models.
\end{abstract}
\section{Introduction}
Finite symbolic abstractions of control systems are used in
algorithmic controller
synthesis~\cite{pessoa,2012-majumdar-approximately,pgt08}.
Since digital implementations of continuous control
systems~\cite{1997-franklin-digital} have quantized and
bounded input space, we consider the setting of bounded and
\emph{quantized-input} control systems.  For such systems, a
state-time quantized abstraction restricted to a compact
region gives a finite abstraction, because the set of input
trajectories is already finite (quantized and
bounded)~\cite{2008-tabuada-approximate}.  The problem of
constructing approximately similar state-time quantized
symbolic abstraction of \emph{possibly unstable}
quantized-input control systems under stabilizability
assumptions has been solved
previously~\cite{2008-tabuada-approximate}.  On the other
hand, the problem of constructing approximately bisimilar
symbolic abstractions of bounded input \emph{unstable}
systems has not been tackled yet.  The difference between an
approximately bisimilar and an approximately similar
abstraction is in the completeness of the abstractions, as
explained in the following.  An \emph{ideal} state-time
quantized symbolic abstraction of a control system would be
exactly bisimilar to the time-quantized system model, but
such exactly bisimilar abstraction is almost impossible to
realize because of symbolic approximations resulting from
quantization of state space.  An exact bisimulation
relationship between a time quantized system model and a
state-time quantized symbolic model can be equivalently
factored into the conjunction of the following three
conditions, which we call \emph{zero deviation, soundness
  and completeness} respectively.
\begin{enumerate}
\item Zero deviation: The deviation between the output of a
  system state and the related symbolic state would
  \emph{ideally} be zero.
\item Soundness: Let $s$ be a state of a time-quantized
  model, and $s_r$ be its related symbolic state in the
  state-time quantized model. Soundness holds if whenever
  $s_r$ transitions to $s_r^\pr$ by some input, then there
  is a corresponding input by which $s$ transitions to
  $s^\pr$ which is symbolically related to $s_r^\pr$.
  The difference between soundness and an exact
    simulation~\footnote{\label{foot:milner}Exact simulation and
      bisimulation are defined in
     Girard and Pappas~\cite{metrics07}.} relation is that
    soundness does not require the outputs of related states
    to be the same, but for an exact simulation relation to
    hold, it is necessary (not sufficient) that the outputs
    of related states are same.
\item Completeness: Completeness is the converse of
  soundness and is defined as follows.  Let $s$ be a state
  of a time-quantized model, and $s_r$ be its related
  symbolic state in the state-time quantized
  model. Completeness holds if whenever $s$ transitions to
  $s^\pr$ by some input, then $s_r$ transitions to $s_r^\pr$
  by a corresponding input such that $s^\pr$ is symbolically
  related $s_r^\pr$.  Completeness also does not
    require the outputs of related states to be same, unlike
    an exact simulation$^{\ref{foot:milner}}$ relation.
\end{enumerate}
Unlike bisimulation, an exact
simulation$^{\ref{foot:milner}}$ of a state-time quantized
model by the time-quantized model only entails zero
deviation plus soundness, but not completeness.  On the
other hand, an exact bisimulation relation between a
state-time quantized model and the time-quantized model
entails completeness, soundness and zero deviation.  The
soundness condition ensures that every control law
synthesized from the symbolic model has its corresponding
control law in the time-quantized system model.  The
completeness condition ensures that all control laws present
in the time-quantized system model have corresponding
control laws in the state-time quantized symbolic model;
which means that we do not miss out any control laws of the
time-quantized system model from the symbolic model while
doing controller synthesis on symbolic model.  The zero
deviation condition ensures that there is no error in the
output of the synthesized control law from symbolic model
when compared with the actual output of the system for the
same control law.

The soundness condition is indispensable because without it
the control laws synthesized from the symbolic model would
not be correct for the actual system model.  But unlike the
soundness condition, the zero deviation and completeness
conditions are not very imperative.  In fact, satisfying the
zero deviation condition is very difficult if not impossible
because of state-quantization induced symbolic
approximations.  So in approximately similar symbolic
abstraction~\cite{pgt08,mgmp12-unstable,2012-tazaki-discrete,2008-tabuada-approximate},
the zero deviation condition is relaxed as:

\emph{Parametrized deviation}: There is a parameter
specifying an upper bound on the deviation between the
states of the system model and related states of the
symbolic model. We will call this parameter as
\emph{proximity}.

The extant methodology of approximately similar symbolic
abstraction (discussed
in~\cite{pgt08,mgmp12-unstable,2012-tazaki-discrete,2008-tabuada-approximate})
establishes soundness between symbolic model and system
model while also specifying the proximity parameter, which
is the precision bound of the approximate simulation
relation~\cite{pgt08,mgmp12-unstable,2012-tazaki-discrete,2008-tabuada-approximate}.
Parametrization of the amount of deviation (proximity)
between a system model and its symbolic model in addition to
demonstrating soundness of the model is the advantage of
approximately similar symbolic abstraction.  Also, a good
methodology of symbolic abstraction allows for adjusting the
{proximity} to a very small amount.  In this regard, the
methodologies discussed
in~\cite{pgt08,mgmp12-unstable,2012-tazaki-discrete,2008-tabuada-approximate}
permit sound symbolic abstraction with arbitrarily small
proximity.  However, a stronger method of abstraction,
discussed in~\cite{pgt08}, can construct an approximately
bisimilar (not just similar) finite symbolic model to a
time-quantized system model of a globally asymptotically
stable system, in which case the abstraction is {complete}
in addition to being sound and proximate.

Although proximity has been parametrized through the notion
of approximate
simulation~\cite{pgt08,mgmp12-unstable,2012-tazaki-discrete,2008-tabuada-approximate},
no attempt has been made until now to parametrize
completeness.  Parametrization of completeness would be
useful while abstracting bounded input unstable systems
that, in many cases, can not have fully complete (plus
sound) state-time quantized symbolic models.  A parameter
for completeness quantifies how exhaustively we can search
for control laws using the symbolic model.  Our paper is
concerned about parametrization of completeness and finding
a way of near-complete, sound, and proximate state-time
quantized abstraction of \emph{bounded and quantized input
  possibly unstable but locally asymptotically stabilizable
  linear control systems}.  We formalize near completeness,
soundness and proximity by the notion of \emph{trimmed input
  approximate bisimulation}, which is introduced in our
paper.  We employ supervisory feedback in the process of
abstraction.  Note that when the input space is bounded,
then locally stabilizable divergent linear systems are still
not globally asymptotically stabilizable (refer to
Subsection~\ref{subsec:enabling} and Appendix).  Therefore
we make the distinction between local asymptotic
stabilizability and global asymptotic stabilizability of
{bounded input} linear systems.  Before we explain our work,
we would like to motivate it by discussing some Related Work
as follows.

\subsection*{Related Work}
For globally asymptotically stable (GAS) continuous control
systems, finite approximately bisimilar symbolic automata
models with arbitrarily small approximation can be
constructed by the procedure discussed in~\cite{pgt08}.  As
such, soundness, completeness and proximity conditions are
met by the symbolic model of a GAS system constructed by the
procedure discussed in~\cite{pgt08}.  Regarding unstable
systems and also those systems that meet a stabilizing
condition~\cite{2012-tazaki-discrete,2008-tabuada-approximate}
but are possibly unstable, there has been work on
abstracting the systems into similar symbolic models
i.e. based on approximate simulation, but not on approximate
bisimulation
~\cite{mgmp12-unstable,2012-tazaki-discrete,2008-tabuada-approximate}.
In other words, these
approaches~\cite{mgmp12-unstable,2012-tazaki-discrete,2008-tabuada-approximate}
for abstracting unstable systems meet the proximity and
soundness conditions, but not the completeness condition.

Since global asymptotic stability (GAS) seems crucial for
complete symbolic abstraction of control systems, it is
tempting to use feedback to globally asymptotically
stabilize the system.  An idea of globally asymptotically
stabilizing the system for symbolic abstraction is discussed
in~\cite{2008-tabuada-approximate}.  But the approach
in~\cite{2008-tabuada-approximate} has the following
discrepancies:
\begin{itemize}
\item The symbolic abstraction procedure
  in~\cite{2008-tabuada-approximate} is concerned about
  sound and proximate abstraction, but not a complete
  abstraction, because the abstraction is based on approximate
  simulation but not on bisimulation.  No explicit attempt
  is made for complete symbolic abstraction.
\item Any linear stabilizing
  supervisory\footnote{\label{foot:super}Supervisory input
    is defined in the Appendix.} input will move out of
  bounds of a bounded input set for some values of the
  original input.  In other words, a supervisory input
  function$^{\ref{foot:super}}$ like $k(y,x,u)=u+C(y-x)$
  will translate the input set by $C(y-x)$ which means
  that the supervisory input moves out of bounds for many
  values of $u$.  If there is a state quantization of
  $\eta$, then there would be a translation of as much as
  $||C||\eta$ between the range of inputs enabled at a
  representative point and a point symbolically approximated
  to a representative point.
\item Global asymptotic stabilization through feedback may
  be possible if the input set is unbounded.  But when an
  everywhere divergent linear system works on a
  {bounded input set}, then the system can not be
  globally asymptotically stabilized (proved in Appendix of
  our paper).  So, the supervisory feedback approach
  in~\cite{2008-tabuada-approximate} can not be directly
  applied for symbolic abstraction of {bounded input}
  everywhere divergent linear systems.
\end{itemize}

\subsection*{Our approach}
Just like the notion of parametrized deviation (or
proximity), it would be beneficial to have a notion of
parametrized completeness since there is no extant method of
constructing fully complete models for bounded input
unstable systems.  Our paper is concerned about
parametrization of completeness in symbolic abstraction by
what we call trimming of input set, where the amount of
completeness is reflected in the smallness of trimming.  We
achieve this by introducing the quintessential idea of
\emph{trimmed input approximate bisimulation}.  We
subsequently discuss a methodology, employing supervisory
feedback, of symbolic abstraction by which we can construct
 {sound} models with arbitrarily small  {proximity
  and trimming}.  The nicety of the procedure of abstraction
is that sound models with arbitrarily small trimming and
proximity can be built, where the trimming is proportional
to the precision bound but \emph{independent of} the time
and state quantization parameters.
Many bounded input unstable linear systems can be locally
asymptotically stabilized.  Therefore our approach can have
significant use.  Although the motivation for this approach
is the idea discussed in~\cite{2008-tabuada-approximate},
but we overcome the drawbacks
of~\cite{2008-tabuada-approximate} mentioned earlier in the
Related work as follows.
\begin{itemize}
\item We parametrize the amount of completeness as smallness
  of trimming of input set.  In our paper, in addition to
  constructing sound and proximate symbolic models, we can
  construct near-complete symbolic models (arbitrarily small
  trimming) bounded and quantized input, possibly unstable,
  locally stabilizable linear systems.  But the completeness
  issue is ignored in~\cite{2008-tabuada-approximate}.
\item Earlier we have stated that a linear stabilizing
  supervisory input can move out of bounds of the input set
  for some values of the original input.  Therefore, we trim
  the input set by a small amount proportional to the
  precision bound while abstracting the symbolic model, such
  that the supervisory input does not go out of bounds of
  the input set. (Section~\ref{subsec:enabling}).
\item The approach in our paper can handle everywhere
  divergent linear systems with {bounded input}, provided
  the system is locally asymptotically stabilizable.  On the
  other hand, the approach
  of~\cite{2008-tabuada-approximate} insists on global
  asymptotic stabilizability.  But everywhere divergent
  linear systems with {bounded input} can not be globally
  asymptotically stabilized as proved in the Appendix of our
  paper.
\end{itemize}
\section{Analog approximation of quantized control systems}\label{sec:analog}
The motivation for our paper is similar
to~\cite{2008-tabuada-approximate} in attempting to build
state-time quantized abstractions of input-quantized control
systems under stabilizability assumptions, but in the scope
of quantized input \emph{linear} control systems.  In this
context, we note the following points about quantized
approximation of inputs.  A quantized input set is generally
an approximation obtained by rejecting noise of the range of
a set of analog input
trajectories~\cite{1958-bertram-effect,1964-slaughter-quantization}.
But if we were to include the noise in inputs, then the
range of input trajectories \emph{without quantization} is
crudely an open subset of an euclidean space.  Therefore in
this paper, instead of directly quantizing state space and
time of the input-quantized control system, we alternatively
obtain a state-time quantized symbolic abstraction of the
\emph{analog (or open input set)} approximation of the
control system, and subsequently restrict the open input set
to the actual quantized input set after the state-time
quantization.  The reason for doing this is because an
open set, which is dense, admits the notion of trimming %
introduced in our paper, which otherwise can not be defined
on discrete sets (this will be explained later in the %
paper).
The supervisory feedback employed in finite
abstraction can also be quantized
(see~\cite{2000-brockett-quantized,1990-delchamps-stabilizing,2005-fu-sector}
about feedback quantization).  Also, a relevant example is
worked out in Section~\ref{sec:example}.
\section{Notation}
\tbf{Important}: In the paper, the word ``open'' refers to
the topological notion of open sets, in the sense that all
points of an open set are interior points.  The word should
not be confused otherwise.  Also, by an open input control
system, we mean that the set of input trajectories of the
control system is topologically open, and this should not be
confused with the extant terminology on supervisory feedback
where open input refers to \emph{non-supervisory} input.

Apart from the general mathematical notations, we use the
following notations.  $\mathbb{R}^+$ refers to the set of
positive real numbers and $\mathbb{R}_{\geq 0}$ is the set
of non-negative real numbers.  $]a,b[$ denotes a left-right
open interval between real numbers $a$ and $b$.  Similarly,
$[a,b]$ is left-right closed, $[a,b[$ is left-closed
right-open interval and $]a,b]$ is right-closed left-open
interval.  We denote $\mathbb{Z}$ as the set of integers and
$\mathbb{N}$ as the set of natural numbers.  If $X$ is a
set, then $X^n=X\times_1X\times_2...\times_n X$. If $x\in
X^n$, then for any $i\in\mathbb{N}$, $x_i$ is the $i^{th}$
component of $x$.  If $\mathbb{R}^n$ is the $n$-dimensional
euclidean space, then for a state quantization parameter
$\eta>0$ we write $[\mathbb{R}^n]_{\eta}=\{x: \exists
k\in\mathbb{Z}^n.x=\eta(k_1,k_2,...,k_n)\}$.  We use the
$L^{\infty}$ norm everywhere in the paper denoted by
$||.||$.  If $X$ is a normed vector space and $I$ being a
connected interval of real line, $\tbf{u}:I\ra X$ and
$\tbf{v}:I\ra X$ are two functions on the same domain $I$,
then the distance norm between them is
$||\tbf{u}-\tbf{v}||=\sup_{t\in
  I}||\tbf{v}(t)-\tbf{u}(t)||.$
\section{Locally asymptotically stabilizable linear control
  systems and trimming of input trajectory set}
\label{sec:finite-abstraction}
A linear control system is a tuple $\Sigma=\left<A_{n\times
    n},B_{n\times m},U,\mathcal{U}\right>$ where $A_{n\times
  n}$ and $B_{n\times m}$ have all real entries, $U\subseteq
\mathbb{R}^m$ and $\mc{U}$ is a subset of all piecewise
continuous input trajectories of the form
$\tbf{u}:[0,\tau]\ra U$, where $\tau>0$ can be any positive
real number.  Additionally, we may also include piecewise
continuous trajectories until infinite time of the form
$\tbf{u}:[0,\infty[\ra U$.  We shall denote $U^{[0,\tau]}$
as the set of all piecewise continuous input trajectories
until time instant $\tau$.

If $x\in \mathbb{R}^n$, then we say that $x$ is a point in
the state space of the linear system $\Sigma$ as above.  An
\emph{absolutely continuous function} $\tbf{x}:[0,\tau]\ra
\mathbb{R}^n$ is said to be a trajectory of the linear
system if there exists $\tbf{u}\in \mc{U}\cap U^{[0,\tau]}$
such that at \emph{almost all} $t\in[0,\tau]$,
$\dot{\tbf{x}}(t)=\frac{d\tbf{x}}{dt}(t)=A\tbf{x}(t)+B\tbf{u}(t)$.
Given the initial condition $\tbf{x}(0)=x$ and an input
trajectory $\tbf{u}\in \mc{U}\cap U^{[0,\tau]}$, the state
trajectory (which is continuous) $\tbf{x}$ driven by
$\tbf{u}$ is uniquely determined.  Then we write
$\tbf{x}(x,t,\tbf{u})$ as the point in state space reached
at time instant $t$ by the trajectory $\tbf{x}$ driven by
$\tbf{u}$.

\subsection{Local asymptotic stabilizability}
A linear system $\Sigma=\left<A_{n\times n},B_{n\times
    m},U,\bigcup_{t\in\mathbb{R}^+}U^{[0,t]}\right>$, where
$U$ is open and bounded and
$\bigcup_{t\in\mathbb{R}^+}U^{[0,t]}$ is the set of all
possible piecewise continuous input trajectories until any
arbitrary time instant, is said to be locally asymptotically stabilizable
if $\forall~(x_{eq},u_{eq})\in\mathbb{R}^n\times U$ satisfying
$Ax_{eq}+Bu_{eq}=0$, there exists an open neighborhood
$Ngh_r=\{y\in\mathbb{R}^n: ||y-x_{eq}||<r\}$ and a matrix
$C_{m\times n}$ such that $\forall y\in Ngh_r$ we have
$(u+Cy)\in U$ and the linear system $\dot{y}=(A+BC)y$ is
asymptotically stable in $Ngh_r$.

{When the input set of a linear control system does not have
  boundaries, then it is well known that global
  stabilizability of linear systems is equivalent to local
  asymptotic stabilizability.  But when the input set of an
  everywhere divergent linear system is bounded, then the
  system can not be globally asymptotically stabilized, as
  proved in the Appendix.  It could still be locally
  asymptotically stabilized and we discuss an example in
  this Section.  Therefore we make the distinction between
  local asymptotic stabilizability and global
  stabilizability of {bounded input} linear systems.}
\begin{rem}\label{rem:matrix-exist}
  $\Sigma$ is locally asymptotically stabilizable if and only if there
  exists a matrix $C$ such that $A+BC$ has all eigenvalues
  with negative real part.  It is well known fact that
  asymptotic stability of a system of linear differential
  equations is equivalent to the prefix matrix of the
  system, in this case $(A+BC)$, having all negative
  eigenvalues.  Then because $U$ is open, the radius $r$ of
  neighborhood $Ngh_r$ can be chosen very small so that
  $\forall y\in Ngh_r.~(u+Cy)\in U$.
\end{rem}
We say that $\Sigma$ as above has a stabilization matrix
$C_{m\times n}$ if $(A+BC)$ has all eigenvalues with
negative real part.  For example, a linear system $$\left[\begin{array}{lcr}\dot{x}_1\\\dot{x}_2 \end{array}\right]=\left[\begin{array}{lcr}
    \mbox{0} & 1\\
    \mbox{-1} & 2\\
  \end{array}\right]\left[\begin{array}{lcr}x_1\\x_2 \end{array}\right]
+ \left[\begin{array}{lcr}\mbox{0}\\\mbox{u} \end{array}
\right]~~u\in]-5,5[$$ where $A=\left[\begin{array}{lcr}
    \mbox{0} & 1\\
    \mbox{-1} & 2\\
  \end{array}\right]$ and $B=\left[\begin{array}{lcr}\mbox{0}\\\mbox{1} \end{array}
\right]$ and $u\in]-5,5[~$ is unstable because
$A=\left[\begin{array}{lcr}
    \mbox{0} & 1\\
    \mbox{-1} & 2\\
  \end{array}\right]$ has both eigenvalues equal to $+1$.
But the system has a stabilization matrix $C=[0~-4]$
because $(A+BC)=\left[\begin{array}{lcr}
    \mbox{0} & 1\\
    \mbox{-1} & 2\\
  \end{array}\right]+\left[\begin{array}{lcr}\mbox{0}\\\mbox{1} \end{array}
\right].[0~-4]=\left[\begin{array}{lcr}
    \mbox{0} & 1\\
    \mbox{-1} & -2\\
  \end{array}\right]$ has both eigenvalues equal to $-1$
which is negative.  In fact, at the equilibrium
$(x_{eq},u_{eq})=(0,0)$ for constant input $u_{eq}=0$, take
a neighborhood $Ball_{1}(0)$ of radius $1$ around $x_{eq}=0$
and then we get that $\forall y\in
Ball_{1}(0).||u_{eq}+Cy||=||0+Cy||\leq
||C||||y||=||[0,-4]||.||y||<4\times 1$.  Hence the feedback
input $(u_{eq}+Cy)$ for $y\in Ball_{1}(0)$ will remain
within the bounded input set $]-5,5[$ for all $y\in
Ball_{1}(0)$.

However, for other values of input $u$, the feedback $u+Cy$
may move out of the bounded space $]-5,5[$.  The supervisory
function $k(y,u)=u+Cy$ translates the bounded input set by
$Cy$ and therefore moves out of the original input set for
some values of $u$.  In other words, we can not use linear
stabilizing feedback directly in symbolic abstraction when
the input set is bounded, because the stabilizing feedback
is specific to a certain constant equilibrium input in a
corresponding state space neighborhood around the state
equilibrium.  Although the example here is locally
stabilizable as shown, but it is not globally
stabilizable. The proof is in the Appendix.
\subsection{Trimming of open sets and corresponding
  trajectory space}
If $S$ is any normed vector space, then for any $s\in S$, we
write an open ball (square) of radius $\rho>0$ around any
point $s$ as ${Ball}_\rho(s)=\{s^\pr\in
S:\|s^\pr-s\|<\rho\}$.  Similarly, a closed ball (square) of
radius $\rho>0$ around any point $s$ as
$\ov{Ball}_\rho(s)=\{s^\pr\in S:\|s^\pr-s\|\leq\rho\}$.
Note that $\ov{Ball}_\rho(s)$ defines a closed square of
side length $\rho$ because $||.||$ is the $L^\infty$ norm.
We define the notion of trimming of any open subset of a
metric space as follows.
\begin{defn}
\label{defn:trim}
Let $S$ be a normed vector space.  Then define, for any
\emph{open set} $A\subset S$, $A_{-\rho}=\{s\in
A:~\ov{Ball}_\rho(s)\subseteq A\}$.  Minus $'-'$ in
subscript of $A_{-\rho}$ means trimming.
\end{defn}
Note that we used a closed ball for trimming and not an open
ball.  We then derive in Proposition~\ref{prop:open} that
the trimmed set of an open set is open.  Before that, we
discuss an example and explain the reason why we defined
trimming on only open sets.

\emph{Example of trimmed set}: Consider a two dimensional
open rectangle $rect=]2,4[\times]9,14[$. Recall that
$\ov{Ball}_{0.3}(.)$ defines a closed square of side length $0.3$
because we are considering $L^\infty$ norm.  Therefore,
after trimming the open rectangle by an amount $\rho=0.3$,
we get an open rectangle $rect_{-0.3}=]2.3,3.7[\times
]9.3,13.7[$ because all closed squares of side length $0.3$
attached to the boundary are removed.

$A_{-\rho}$ is the set obtained by trimming $A$ by a margin
of $\rho$ near the boundary, since all points near the
boundary within a margin of $\rho$ have at least one point
among their $\rho$-distant neighbors outside $A$.  Although
the definition of trimming can also extend to non-open sets,
but in a practical sense, trimming is more reasonable for
open sets.  To illustrate, consider a finite but large
subset of a normed vector space.  Then trimming of the finite set
by even an infinitesimally small margin will result in an
empty set, because all the points of the finite set are
boundary points.  To avoid this oddity, we restrict the
definition of trimming to open sets only.
The following proposition asserts that after trimming an
open set, we end up with an open set.
\begin{prop}\label{prop:open}
  If $A$ is an open subset of a Banach space $S$, then for
  any $\rho>0$, $A_{-\rho}$ is open.  
  \end{prop}
  Note that the above Proposition~\ref{prop:open} does not
  hold if in the Definition~\ref{defn:trim} of trimmed set ,
  the closed ball used for trimming is replaced by an open
  ball.
\begin{proof}
  Take any point $a\in A_{-\rho}$.  This means
  $\ov{Ball}_{\rho}(a)\subset A$ by the definition of
  trimming.  For $w>\rho$, define $X=Ball_w(a)\cap A$.  As
  $w>a$, so $\ov{Ball}_\rho(a)\subset Ball_w(a)$.  Also we
  have $\ov{Ball}_\rho(a)\subset A$.  Therefore,
  $\ov{Ball}_\rho(a)\subset X=Ball_w(a)\cap A$.
  $\ov{Ball}_\rho(a)$ is a strict subset of $X$ because
  $\ov{Ball}_\rho(a)$ is closed while $X$ is open.  Let
  $\partial X$ be the boundary of $X$, which is compact
  because $X$ is bounded and $S$ is Banach.  As
  $\ov{Ball}_\rho(a)\subset X$ and $X\cap \partial X=\{\}$-
  $X$ being open, so $\partial X\cap\ov{Ball}_\rho(a)\subset
  X\cap \partial X=\{\}$.  Since $\partial X\cap
  Ball_\rho(a)=\{\}$, so for every $x\in\partial X$, we have
  $||a-x||>\rho$ and therefore we can choose
  $r_x:0<r_x<\left(||a-x||-\rho\right)$ and $\delta_x:0
  <\delta_{x}<\left(||a-x||-\rho-r_x\right)$.  By reverse
  triangular inequality, if $y\in Ball_{\delta_x}(x)$, then
  $||y-a||>||a-x||-||y-x||>||a-x||-\delta_x$.  Since
  $\delta_x$ is chosen such that
  $\delta_{x}<\left(||a-x||-\rho-r_x\right)$, so by
  substituting we get $||y-a||>r_x+\rho$.  Therefore all
  points $y\in Ball_{\delta_x}$ are at a distance of greater
  than $\rho$ from $a$ and so $Ball_{\delta_x}\cap
  Ball_{\rho+r_x}(a)=\{\}$.  Consider the open cover of
  $\partial X$ as $Cov=\{Ball_{\delta_x}(x):x\in\partial
  X\}$.  Because $\partial X$ is compact, so there exists a
  finite sub-cover $FinCov\subset Cov$ covering $\partial X$.
  Index sets in $FinCov$ as
  $FinCov=\{Ball_{\delta_{x1}}(x1),Ball_{\delta_{x2}}(x2),...,Ball_{\delta_{xk}}(xk)\}$
  for some $k\in\mathbb{N}$.  Let $r=\min_{1\leq i\leq
    k}r_{xi}$ where $r_{x}$ is chosen for any $x\in\partial
  X$ as described earlier.  We earlier showed that $\forall
  x\in\partial X.Ball_{\delta_x}\cap
  Ball_{\rho+r_x}(a)=\{\}$ which means $Ball_{r+\rho}(a)$ is
  disjoint from each of the sets in $FinCov$ which covers
  $\partial X$ since $r=\min_{1\leq i\leq
    k}r_{xi}$.  Therefore $Ball_{r+\rho}(a)$ is disjoint
  from $\partial X$.

  We now show that $Ball_{r+\rho}(a)\subseteq X$.  $X$ is open
  and also $[complement(X)/\partial X]$ is open by removing
  the boundary from $complement(X)$.  We have
  $Ball_{r+\rho}(a)=\left(Ball_{r+\rho}(a)\cap
    X\right)\cup\left(Ball_{r+\rho}(a)\cap
    [complement(X)/\partial
    X]\right)\cup\left(Ball_{r+\rho}(a)\cap\partial
    X\right)$.  But earlier we proved
  $Ball_{r+\rho}(a)\cap\partial X=\{\}$ by which we
  get $$Ball_{r+\rho}(a)=\left(Ball_{r+\rho}(a)\cap
    X\right)\cup\left(Ball_{r+\rho}(a)\cap
    [complement(X)/\partial X]\right)$$
  $\left(Ball_{r+\rho}(a)\cap X\right)$ and
  $\left(Ball_{r+\rho}(a)\cap [complement(X)/\partial
    X]\right)$ are both open.  $S$ being a Banach space, all
  balls in the space are connected and so $Ball_{r+\rho}(a)$
  is connected and can not be written as the disjoint union
  of two open sets.  Therefore either
  $\left(Ball_{r+\rho}(a)\cap X\right)$ or
  $\left(Ball_{r+\rho}(a)\cap [complement(X)/\partial
    X]\right)$ is empty.  Also earlier we proved
  $Ball_{r+\rho}(a)\cap\partial X=\{\}$.  As
  $\ov{Ball}_\rho(a)\subset Ball_{\rho+r}(a)$ and
  $\ov{Ball}_\rho(a)\subset X$, so $Ball_{\rho+r}(a)\cap X$
  is non-empty.  This means the other open set in disjoint
  union $Ball_{r+\rho}(a)\cap[complement(X)/\partial
  X]=\{\}$ (empty).  So, $\left(Ball_{r+\rho}(a)\cap
    X\right)=Ball_{r+\rho}(X)$ or equivalently
  $Ball_{r+\rho}(a)\subseteq X$.

  Consider any $p\in Ball_r(a)$.  Then, $\forall q\in
  \ov{Ball}_\rho(p)$, we have by triangular inequality
  $||q-a||\leq ||p-a||+||p-q||<r+\rho$ substituting
  $||p-a||<r$ while $||p-q||\leq \rho$.  So, $q\in
  Ball_{\rho+r}(a)$.  But, $Ball_{\rho+r}(a)\subseteq
  X\subseteq A$ implies $q\in A$.  Therefore there exists an
  open neighborhood as $Ball_r(a)$ around $a$ such that
  $\forall p\in Ball_r(a)$, $Ball_\rho(p)\subseteq A$ or
  equivalently $Ball_r(a)\subseteq A_{-\rho}$ by the
  definition of trimming.  Without loss of generality, for
  any $a\in A_{-\rho}$, we can find a corresponding $r>0$
  such that $Ball_r(a)\subseteq A_{-\rho}$.  Therefore
  $A_{-\rho}$ is open.
  \end{proof}
  We can define trimming on an open set of all piecewise
  continuous input trajectories in either of the following
  two ways 1) Trim the co-domain of the trajectories and
  then define piecewise continuous input trajectories on the
  trimmed co-domain.  2) Trim the actual set of input
  trajectories.  The Proposition~\ref{prop:trimming} asserts
  that both the above ways of trimming result in the same
  set.  For example, we know that
  $\left(]0,3[\right)_{-0.1}=]0.1,2.9[$.  Then
  $\left((]0,3[)^{[0,1]}\right)_{-0.1}=(]0.1,2.9[)^{[0,1]}$
  where the exponent $[0,1]$ is the time interval for the
  trajectories and subscript $-0.1$ is the amount of
  trimming (minus '-' denotes trimming).
\begin{prop}\label{prop:trimming} 
  Let $U$ be open subset of a normed vector space $S$.  It is
  easy to see that $\forall\tau>0$, $U^{[0,\tau]}$ is also
  open subset of $S^{[0,\tau]}$.  Then $\forall\rho>0$ we have
  $\left(U^{[0,\tau]}\right)_{-\rho}=\left(U_{-\rho}\right)^{[0,\tau]}$.
\end{prop}
\begin{proof}
  Recall that we defined $S^{[0,\tau]}$ to be the set of all
  piecewise continuous input trajectories of the form
  $\tbf{u}:[0,\tau]\ra S$.  We leave it to the reader to
  verify that given $U$ is open in $S$, we have
  $U^{[0,\tau]}$ as also open in $S^{[0,\tau]}$.  We prove
  the main part of the proposition as follows.

  First we prove
 $\left(U^{[0,\tau]}\right)_{-\rho}\subseteq\left(U_{-\rho}\right)^{[0,\tau]}$.
  Let $\tbf{u}\in \left(U^{[0,\tau]}\right)_{-\rho}$.  Then
  we have to prove that $\forall t\in[0,\tau].\tbf{u}(t)\in
  U_{-\rho}$.  For any $t^\pr\in[0,\tau]$ and for any
  $v:||\tbf{u}(t^\pr)-v||<\rho$ define $\tbf{v}:[0,\tau]\ra
  U$ as $$\left|\begin{split}& \tbf{v}(t^\pr)=v\\
      & \tbf{v}(t)=\tbf{u}(t)~~\text{if}~t\neq t^\pr
\end{split}\right.$$  Since $\tbf{u}$ and $\tbf{v}$
differ at only one time point $t^\pr$ so
$||\tbf{u}-\tbf{v}||=||\tbf{u}(t^\pr)-\tbf{v}(t^\pr)||<\rho$.
This means $\tbf{v}\in \ov{Ball}_\rho(\tbf{u},U^{[0,\tau]})$.  As
$\tbf{u}\in\left(U^{[0,\tau]}\right)_{-\rho}$, so
$\tbf{v}\in U^{[0,\tau]}$.  This means $v=\tbf{v}(t^\pr)\in
U$.  But $v$ is any point inside $\ov{Ball}_\rho(\tbf{u}(t),U)$.
So, $\forall v\in \ov{Ball}_\rho(\tbf{u}(t),U)$ we get $v\in U$.
So, $\tbf{u}(t)\in U_{-\rho}$.  This is true for all
$t\in[0,\tau]$.  Therefore $\tbf{u}\in
[U_{-\rho}]^{[0,\tau]}$.  This proves
\begin{equation}\label{eqn:forward}
\left(U^{[0,\tau]}\right)_{-\rho}\subseteq\left(U_{-\rho}\right)^{[0,\tau]}. 
\end{equation}
We shall now prove the converse
$\left(U_{-\rho}\right)^{[0,\tau]}\subseteq
\left(U^{[0,\tau]}\right)_{-\rho}$.  Let $\tbf{u}\in
\left(U_{-\rho}\right)^{[0,\tau]}$.  Take any
$\tbf{v}:||\tbf{u}-\tbf{v}||<\rho$.  Then $\forall
t\in[0,\tau].\tbf{v}(t)\in \ov{Ball}_\rho(\tbf{u}(t),U)$.
$\tbf{u}\in \left(U_{-\rho}\right)^{[0,\tau]}$ means $\forall
t\in[0,\tau]\tbf{u}(t)\in U_{-\rho}$.  So, $\forall
t\in[0,\tau]\tbf{u}(t)\in U_{-\rho}$ and $\forall
t\in[0,\tau].\tbf{v}(t)\in \ov{Ball}_\rho(\tbf{u}(t),U)$ means
$\forall t\in[0,\tau].\tbf{v}(t)\in U$ or equivalently
$\tbf{v}\in U^{[0,\tau]}$.  So, $\forall\tbf{v}\in
Ball_{-\rho}(\tbf{u},U^{[0,\tau]}).\tbf{v}\in U^{[0,\tau]}$.
Therefore $\tbf{u}\in\left(U^{[0,\tau]}\right)_{-\rho}$.
This means that 
\begin{equation}\label{eqn:backward}\left(U_{-\rho}\right)^{[0,\tau]}\subseteq
  \left(U^{[0,\tau]}\right)_{-\rho}.\end{equation}

From (\ref{eqn:forward}) and (\ref{eqn:backward}) we get
  that $\left(U^{[0,\tau]}\right)_{-\rho}=\left(U_{-\rho}\right)^{[0,\tau]}$.
  \end{proof}
 \subsection{Enabling of asymptotically stabilizing
  supervisory inputs}\label{subsec:enabling}
If $\tbf{x}:[0,\infty[\ra X$ is a
trajectory of the linear system $\Sigma$ and $C_{m\times n}$
is a real matrix, then write
$\tbf{y}_{C,\tbf{x}}:[0,\infty[\ra X$ satisfying
$\tbf{y}_{C,\tbf{x}}(0)=y$ and
\begin{equation}\label{eqn:rate-difference}
\left(\dot{\tbf{y}}_{C,\tbf{x}}-\dot{\tbf{x}}\right)(t)=(A+BC)(\tbf{y}_{C,\tbf{x}}(t)-\tbf{x}(t)).
\end{equation}
Notice that if $\tbf{x}$ is driven by an input trajectory
$\tbf{u}$, then $\tbf{y}_{C,\tbf{x}}$ is driven by a
supervisory input trajectory in (\ref{eqn:supervisory}) but
\emph{only until} any time $\tau$ such that
$\tbf{u}_{y,x}([0,\tau])\subseteq U$ because $U$ is bounded;
in other words the image of $[0,\tau]$ by the supervisory
input trajectory $\tbf{u}_{y,x}$ has to be inside $U$ where
$\tbf{u}_{y,x}$ is defined as follows.
\begin{equation}\label{eqn:supervisory}
\tbf{u}_{y,x}=\tbf{u}+C.\left(\tbf{y}_{C,\tbf{x}}-\tbf{x}\right)
\end{equation}
For certain values of $\tbf{u}(t)$, $\tbf{u}_{y,x}(t)$ may move
out of the bounded input set, because $\tbf{u}_{y,x}(t)$ is
the translation of $\tbf{u}(t)$ by an amount
$C.\left(\tbf{y}_{C,\tbf{x}}-\tbf{x}\right)(t)$.  We say that
$\tbf{u}$ admits $\tbf{u}_{y,x}$ of (\ref{eqn:supervisory})
\emph{until time} $\tau$ at point $y$ with reference to $x$
if $\tbf{u}_{y,x}|_{[0,\tau]}\in U^{[0,\tau]}$ or
equivalently $\tbf{u}_{y,x}([0,\tau])\subseteq U$.  This is
said because $\tbf{u}_{y,x}$ may move out of the bounded
input set $U$ at some time instant greater than $\tau$.

\begin{defn} We denote $In_C(y,x,\tau)=\{\tbf{u}\in
  U^{[0,\tau]}:~\tbf{u}_{y,x}([0,\tau])\subseteq U\}$ which
  means that $In_C(y,x,\tau)$ contains all input
  trajectories that admit corresponding supervisory input
  trajectories of the form (\ref{eqn:supervisory}) until
  time $\tau$ at point $y$ taking $x$ as the reference.
\end{defn}
\begin{thm}\label{thm:enabled}
  Let $\Sigma=\left<A_{n\times n},B_{n\times
      m},U,\bigcup_{\tau\in\mathbb{R}^+}U^{[0,\tau]}\right>$
  be an open input ($U$ is open set) locally asymptotically stabilizable
  linear system with a stabilization matrix $C_{m\times n}$.
  Let two points $y~\text{and}~x$ in state space be such
  that for some $\epsilon>0$, $||y-x||\leq\epsilon$.  Then
  the following hold
\begin{enumerate}
\item
  $\forall\delta,\tau>0$,
  $\left(U_{-(||C||\epsilon+\delta)}\right)^{[0,\tau]}\subseteq
  In_C(y,x,\tau)$. 
\item $\left(U_{-||C||\epsilon}\right)^{[0,\tau]}\subseteq
  In_C(y,x,\tau)$.  Equivalently
  $\left(U^{[0,\tau]}\right)_{-||C||\epsilon}\subseteq
  In_C(y,x,\tau)$ because
  $\left(U^{[0,\tau]}\right)_{-||C||\epsilon}=\left(U_{-|C||\epsilon}\right)^{[0,\tau]}$
  by Proposition~\ref{prop:trimming}.
\end{enumerate}
\end{thm}
\begin{proof}
  1. We prove the first part of the theorem by
  contradiction.  Assume that there are $\delta,\tau>0$ such
  that
  $\left(U_{-(||C||\epsilon+\delta)}\right)^{[0,\tau]}\nsubseteq
  In_C(y,x,\tau)$.  This means that there exists an input
  $\tbf{u}\in
  \left(U_{-(||C||\epsilon+\delta)}\right)^{[0,\tau]}$ such
  that the image of $[0,\tau]$ by the corresponding
  supervisory input trajectory $\tbf{u}_{y,x}$ is not
  contained inside $U$, i.e. $\tbf{u}_{y,x}([0,\tau])\notin
  U$.  Define $\mathcal{F}=\{t\geq
  0:\tbf{u}_{y,x}([0,t])\notin U\}$.  The set $\mathcal{F}$
  is non-empty because $\tau\in\mathcal{F}$.  Let
  $\omega=\inf{\mathcal{F}}$.  This means that at the
  precise time instant $\omega$, the supervisory input
  trajectory $\tbf{u}_{y,x}$ is at the boundary point of the
  set $U$, i.e. $\tbf{u}_{y,x}(\omega)$ is at the boundary
  point of $U$.

  By (\ref{eqn:supervisory}) we get that
  $\tbf{u}_{y,x}(\omega)=\tbf{u}(\omega)+C.\left(\tbf{y}_{C,\tbf{x}}(\omega)-\tbf{x}(\omega)\right)$
  and from this
  $||\tbf{u}_{y,x}(\omega)-\tbf{u}(\omega)||\leq
  ||C||||\tbf{y}_{C,\tbf{x}}(\omega)-\tbf{x}(\omega)||$.
  But from (\ref{eqn:rate-difference}) we get that
$$\left(\tbf{y}_{C,\tbf{x}}(\omega)-\tbf{x}(\omega)\right)=\exp((A+BC)\omega)(y-x)$$
As $(A+BC)$ has all eigenvalues with negative real part, so
from previous equation we get that
$||\tbf{y}_{C,\tbf{x}}(\omega)-\tbf{x}(\omega)||\leq||y-x||\leq\epsilon$.
Substituting this in what we got earlier, we have
\begin{equation}\label{eqn:epsilon-bound}
  ||\tbf{u}_{y,x}(\omega)-\tbf{u}(\omega)||\leq
  ||C||\epsilon
\end{equation} 
Using this we proceed to prove that $\tbf{u}_{y,x}$ is an
interior point of $U$ which shall be a contradiction to an
earlier conclusion that $\tbf{u}_{y,x}$ is a boundary point
of $U$.  

Consider a closed ball (square)
$\ov{Ball}_\delta(\tbf{u}_{y,x}(\omega))$ of radius $\delta$
around $\tbf{u}_{y,x}(\omega)$.  Consider any point $p\in
Ball_{\delta}(\tbf{u}_{y,x}(\omega))$.  Then
$$||p-\tbf{u}(\omega)||\leq
||p-\tbf{u}_{y,x}(\omega)||+||\tbf{u}_{y,x}(\omega)-\tbf{u}(\omega)||.$$
By substituting from (\ref{eqn:epsilon-bound}) we get that
\begin{equation}\label{eqn:radius}
||p-\tbf{u}(\omega)||\leq
||p-\tbf{u}_{y,x}(\omega)||+||C||\epsilon<\delta+||C||\epsilon
\end{equation}
because $p\in \ov{Ball}_{\delta}(\tbf{u}_{y,x}(\omega))$.

But $\tbf{u}\in
\left(U_{-(||C||\epsilon+\delta)}\right)^{[0,\tau]}$ implies
$\tbf{u}(\omega)\in U_{-(||C||\epsilon+\delta)}$, and then
from (\ref{eqn:radius}) we get that $p\in U$.  This is true
for all $p\in \ov{Ball}_\delta(\tbf{u}_{y,x}(\omega))$ which
means that $\tbf{u}_{y,x}(\omega)$ is an interior point of
$U$.  This is contrary to an earlier conclusion that
$\tbf{u}_{y,x}(\omega)$ is the boundary point of $U$.

This means that the assumption we started with at the
beginning is false.  Hence, $\forall\delta,\tau>0$.
$\left(U_{-(||C||\epsilon+\delta)}\right)^{[0,\tau]}\subseteq
In_C(y,x,\tau)$.

2. The proof of second part of the Proposition is as
follows.  We shall first prove
$\bigcup_{\delta>0}U_{-(||C||\epsilon+\delta)}=U_{-||C||\epsilon}$.
Let $u\in U_{-||C||\epsilon}$.  Since $U_{-||C||\epsilon}$
is open by Proposition~\ref{prop:open}, so
$\exists\delta_u>0.\ov{Ball}_{\delta_u}(u)\subseteq
U_{-||C||\epsilon}$.  By simple geometry, given $u$ is
inside the $||C||\epsilon$-trimmed set $U_{-||C||\epsilon}$,
we get that $\ov{Ball}_{\delta_u+||C||\epsilon}(u)\subseteq
U$.  Equivalently $u\in U_{-(||C||\epsilon+\delta_u)}$.  So,
for every $u\in U_{-||C||\epsilon}$, there exists
$\delta_u>0$ such that $u\in U_{-(||C||\epsilon+\delta_u)}$.
Therefore $U_{-||C||\epsilon}\subseteq
\bigcup_{\delta>0}U_{-(||C||\epsilon+\delta)}$.  Also,
$\bigcup_{\delta>0}U_{-(||C||\epsilon+\delta)}\subseteq
U_{-||C||\epsilon}$ because $\forall\delta>0$,
$U_{-(||C||\epsilon+\delta)}\subseteq U_{-||C||\epsilon}$.
Therefore by sandwitching we get,
$\bigcup_{\delta>0}U_{-(||C||\epsilon+\delta)}=U_{-||C||\epsilon}$.

From the theorem statement, $\forall\delta,\tau>0$,
$\left(U_{-(||C||\epsilon+\delta)}\right)^{[0,\tau]}\subseteq
In_C(y,x,\tau)$.  As the theorem holds for all $\delta>0$
independently of $\tau$, therefore we get that
$\bigcup_{\delta>0}\left(U_{-(||C||\epsilon+\delta)}\right)^{[0,\tau]}\subseteq
In_C(y,x,\tau)$.  Earlier we proved
$\bigcup_{\delta>0}U_{-(||C||\epsilon+\delta)}=U_{-||C||\epsilon}$.
Therefore, $\bigcup_{\delta>0}\left(U_{-||C||\epsilon}\right)^{[0,\tau]}\subseteq
In_C(y,x,\tau)$.  Equivalently by Proposition~\ref{prop:trimming},
$\left(U^{[0,\tau]}\right)_{-||C||\epsilon}\subseteq
In(y,x,\tau)$ .
  \end{proof}
  \section{Trimmed input approximate bisimulation}
\label{sec:tiab}
We define a metric transition system (MTS) as follows.
\begin{defn}[Metric transition system]$~~~$\\
  A Metric Transition System (MTS) is
  $T=\left<X,V,\ra,Y,H\right>$ where $X$ is a state space,
  $V$ is the superset of all possible inputs at any point,
  $[\ra]\subseteq X\times V\times X$ is the transition
  relation, $Y$ is a metric space and $H:X\ra Y$ is the
  output map.
\end{defn}
Since trimming is only defined on open sets
(Definition~\ref{defn:trim}), therefore we identify an
\emph{Open Input Metric Transition System} (OIMTS) as
follows, from which a trimmed open input metric transition
system may be derived after trimming the input set.
\begin{defn}[Open Input Metric Transition System]$~~~$\\
  An Open Input Metric Transition System (OIMTS) is an MTS~
  $T=\left<X,V,\ra,Y,H\right>$ with the additional condition
  that $V$ is an open subset of a normed vector space.
\end{defn}
Related to control systems, the set $V$ in an OIMTS consists
of any open set of input trajectories.
\subsection{Approximate bisimulation without trimming}
We first define one version of $\epsilon$-approximate
simulation according to~\cite{2008-tabuada-approximate}
which is useful when supervisory
feedback$^{\ref{foot:supervisory}}$ is used in symbolic
abstraction.  A different and more common version of
approximate simulation is defined in~\cite{metrics07} but the
definition in ~\cite{metrics07} is not suitable when supervisory
feedback$^{\ref{foot:supervisory}}$ is introduced in
symbolic abstraction, as will be explained in
Remark~\ref{rem:supervisory}.

The author
of~\cite{2008-tabuada-approximate} actually defines a
stronger $\epsilon-\delta$ approximate simulation, having a
$\delta$-reflexivity condition.  But we restrict to the
general $\epsilon$-approximate simulation leaving
$\delta$-reflexivity.
\begin{defn}\label{defn:simulation}
  Let $T=\left<X,V,\ra_1,Y,H\right>$ and
  $T^\pr=\left<X^\pr,V^\pr,\ra_2,Y,H^\pr\right>$ be two MTS.
  Let $Y$ be equipped by the metric $d:Y\times
  Y\ra\mathbb{R}_{\geq 0}$.  Note that the output range $Y$
  is same for both the MTS but the input sets $V$ and
  $V^\pr$ may be different.  We say that a non-empty
  relation $R\subset X\times X^\pr$ is an
  $\epsilon$-approximate simulation relation of $T$ by
  $T^\prime$, iff $\forall (x,x^\pr)\in R$ all the following
  hold
\begin{enumerate}
\item $d\left(H(x),H^\pr(x^\pr)\right)\leq\epsilon$.
\item $\forall y\in X~\wedge~\tbf{u}\in V$, if
  $x\stackrel{\tbf{u}}{\ra}_1y$ then there exist $y^\pr\in
  X^\pr$ and $\tbf{u}^\pr\in V^\pr$ such that
  $x^\prime\stackrel{\tbf{u}^\pr}{\ra}_2y^\prime~\text{and}~(y,y^\prime)\in
  R$.  Note that $\tbf{u}$ and $\tbf{u}^\pr$ may be
  different.
\end{enumerate}  
\end{defn}
\tbf{Approximate bisimulation:} Consequently, we say that
$T$ and $T^\prime$ are $\epsilon$-bisimilar to each other
iff $\exists$ a non-empty relation $R$ such that $R$
$\epsilon$-approximately simulates $T$ by $T^\prime$ and
$R^{-1}$ $\epsilon$-approximately simulates $T^\prime$ by
$T$.
\begin{rem}\label{rem:supervisory}
  The more common definition of $\epsilon$-approximate
  simulation in~\cite{metrics07} requires that one
  transition may simulate another only if both the
  transitions are driven by the same input.  But when
  supervisory feedback~\footnote{\label{foot:supervisory}
    The general definition of supervisory feedback function
    is given in the Appendix, but in the paper we shall only
    discuss locally asymptotically stabilizing linear
    supervisory feedback.} is used, then it may happen that
  one transition simulated by another transition is such
  that the driving input of former transition is a feedback
  supervisory function$^{\ref{foot:supervisory}}$ of the
  input of latter transition and may not be equal to the
  input of latter transition.  As such the definition
  in~\cite{metrics07} is restrictive in the sense that it
  can not be used to analyze symbolic abstraction involving
  supervisory feedback.  On the other hand, the
  Definition~\ref{defn:simulation} of our paper, which is
  also previously stated in~\cite{2008-tabuada-approximate},
  allows us to interpret symbolic abstraction involving
  supervisory feedback.
\end{rem}

\subsection{Trimmed-input approximate bisimulation for
  OIMTS}

\begin{defn}[Trimmed open input metric transition system]\label{defn:toimts}$~~~$\\
  Let $T=\left<X,V,\ra,Y,H\right>$ be a open input metric
  transition system (OIMTS).  Then we define the $\rho$
  trimmed transition system
  $T_{-\rho}=\left<X,V_{-\rho},\xrightarrow[]{}|_{V_{-\rho}},Y,H\right>$
  where $V_{-\rho}$ is the obtained after trimming the open
  set $V$ by $\rho$ near the boundary and
  $\xrightarrow[]{}|_{V_{-\rho}}$ is the restriction of the
  original transition relation to $V_{-\rho}$.
 \end{defn}
\begin{defn}[Trimmed input approximate
  simulation]$~~~$\\
  Let $T=\left<X,V,\ra_1,Y,H\right>$ and
  $T=\left<X^\pr,V^\pr,\ra_2,Y,H\right>$ be two OIMTS.  For
  any $\rho,\epsilon>0$, we say that a relation $R\subset
  X\times X^\pr$ is a $\rho$-trimmed $\epsilon$-approximate
  simulation of the OIMTS $T$ by $T^\prime$ iff the
  relation $R$ is an $\epsilon$-approximate simulation of
  $T_{-\rho}$ by $T^\prime$, where $T_{-\rho}$ is the
  $\rho$-trimmed OIMTS obtained from $T$.
\end{defn}

\begin{defn}[Trimmed-input approximate bisimulation]$~~~$\\
  Consequently, $R$ is a $\rho$-trimmed
  $\epsilon$-approximate bisimulation between $T$ and
  $T^\pr$ iff $R$ $\epsilon$-approximately simulates
  $T_{-\rho}$ by $T^\prime$ and $R^{-1}$
  $\epsilon$-approximately simulates $T^\pr_{-\rho}$ by $T$.
\end{defn}
\section{Near completeness and interpretation of trimmed
  input approximate bisimulation}
\label{sec:interpretation}
We define near completeness as follows.
\begin{defn}\label{defn:near-completeness}
  For any $\gamma>0$, we say that an OIMTS $\widehat{T}$ is
  $\gamma$-near complete with respect to an OIMTS $T$ iff
  there exist $\alpha,\beta>0$ and an OIMTS $T^\pr$ such
  that all the following hold (i) $\alpha+\beta=\gamma$ (ii)
  $T^\pr_{-\beta}=\widehat{T}$ (iii) $T_{-\alpha}$ is
  (approximately) simulated by $T^\pr$.
\end{defn}
{Let $T$ and $T^\pr$ be two OIMTS such that they are
  $\rho$-trimmed input $\epsilon$-approximately bisimilar.}
Then we make the following interpretations about the $\rho$
trimmed transition system $T_{-\rho}^\pr$.  
\begin{itemize}

\item \tbf{$\epsilon$-Proximity}: The distance between two
  related states of $T$ and $T_{-\rho}^\pr$ is less than
  $\epsilon$ since $T$ and $T^\pr$ are $\rho$-trimmed input
  $\epsilon$-approximately bisimilar.\\
\item \tbf{Soundness}: $T^\pr_{-\rho}$ is
  $\epsilon$-approximately simulated by $T$ since $T$ and
  $T^\pr$ are $\rho$-trimmed input $\epsilon$-approximately
  bisimilar.  This means that $T^\pr_{-\rho}$ is sound
  with respect to $T$.

\emph{\underline{Disambiguation}}. It is to be noted that $T^\pr$ may
not be sound with respect to $T$.  Instead we demonstrated
that $T^\pr_{-\rho}$ is sound with respect to $T$.\\
\item \tbf{$2.\rho$-Near completeness}: Since $T$ and
  $T^\pr$ are $\rho$-trimmed input $\epsilon$-approximately
  bisimilar, so $T_{-\rho}$ is $\epsilon$-approximately
  simulated by $T^\pr$.  On the other hand, $T^\pr_{-\rho}$
  is obtained after further trimming the input set of
  $T^\pr$ by $\rho$.  Therefore, by the
  Definition~\ref{defn:near-completeness}, we get that
  $T^\pr_{-\rho}$ is $2\rho$-near complete with respect to
  $T$ by substituting $\alpha=\beta=\rho$ where $\alpha$ and
  $\beta$ are the parameters stated in
  Definition~\ref{defn:near-completeness}.
\end{itemize}
\emph{In a vague sense, when $\rho$ and $\epsilon$ are very
  small, then something like $T^\pr_{-\rho}$ transpires as a
  {reasonably good} abstraction of $T$ after establishing
  $\rho$-trimmed input $\epsilon$-approximate bisimulation
  between $T$ and $T^\pr$.}
\section{State and time quantization}
\begin{defn}[Time quantized transition relation of control
  system]\label{defn:control-transition}$~~$
  For a linear control system $\Sigma$ and two points $x$
  and $y$ in state space and $\tbf{u}\in U^{[0,\tau]}$, we
  write $x\xrightarrow{\tbf{u}}y$ iff
  $\tbf{x}(x,\tau,\tbf{u})=y$.  For any $x$ in state space
  and $\tbf{u}\in U^{[0,\tau]}$, define
  $Reach(x,\xrightarrow{\tbf{u}})=\{y\in\mathbb{R}^n:x\xrightarrow{\tbf{u}}y\}$
\end{defn}

Let an open input linear control system
$\Sigma=\left<A_{n\times n},B_{n\times
    m},U,\bigcup_{t\in\mathbb{R}^+}U^{[0,t]}\right>$ where
$U$ is open. Then for any $\tau>0$, the time quantized
open input metric transition system (OIMTS)
$T^\tau(\Sigma,X)$ is defined as
$$T^\tau(\Sigma)=\left<\mathbb{R}^n,U^{[0,\tau]},\longrightarrow,\mathbb{R}^n,id\right>$$
{where $id$ is the identity output map} and
$\longrightarrow$ is the transition relation according to
Definition~\ref{defn:control-transition}.  We leave it to
the reader to verify that since $U$ is open, so
$U^{[0,\tau]}$ is also open and hence the MTS $T^\tau(\Sigma)$ is
also an OIMTS.  Therefore, $T^\tau(\Sigma)$ shall admit the
notion of trimming.

For any $\eta>0$ and $\tbf{u}\in U^{[0,\tau]}$, we define a
transition relation $\xrightarrow[\eta]{\tbf{u}}\subset
[\mathbb{R}^n]_\eta\times [\mathbb{R}^n]_\eta$ as follows.
$x\xrightarrow[\eta]{\tbf{u}}y$ if and only if $\exists
y^\pr\in Reach(x,\xrightarrow{\tbf{u}}).
||y-y^\pr||\leq\eta/2$.  Then the state-time quantized
OIMTS $T^{\tau,\eta}(\Sigma)$ is defined as
$$T^{\tau,\eta}(\Sigma)=\left<[\mathbb{R}^n]_\eta,U^{[0,\tau]},\xrightarrow[\eta]{},\mathbb{R}^n,id\right>.$$

\begin{cor} \label{cor:time-stamp} Let
  $\Sigma=\left<A_{n\times n},B_{n\times
      m},U,\bigcup_{t\in\mathbb{R}^+}U^{[0,t]}\right>$ be a
  locally asymptotically stabilizable linear control system with open and
  bounded input set $U$ and a stabilization matrix $C$.
  Then $\forall \epsilon>0$ and $\eta:0<\eta<\epsilon$,
  there exists $\tau>0$ such that
  $||\epsilon\exp((A+BC)\tau)||<\eta/2$.  For such $\tau$,
  if $y,x: ||y-x||\leq\epsilon$ and $\tbf{u}\in
  (U_{-||C||\epsilon})^{[0,\tau]}$ and
  $x\xrightarrow{\tbf{u}} x^\pr$, then $\exists
  y^\pr.y\xrightarrow{\tbf{u}_{y,x}}y^\pr$ satisfying
  $||y^\pr-x^\pr||<\eta/2$.
\end{cor}
\begin{proof}
  Firstly, we have to show that for any chosen
  $\tau,\epsilon>0$, if $\tbf{u}\in
  (U_{-||C||\epsilon})^{[0,\tau]}$ and
  $||y-x||\leq \epsilon$, then $\tbf{u}_{y,x}$ is enabled
  until the chosen $\tau>0$.  If $\tbf{u}\in
  (U_{-||C||\epsilon})^{[0,\tau]}$ and
  $||y-x||\leq\epsilon$, then by Theorem~\ref{thm:enabled}
  we get that $\tbf{u}\in In(y,x,\tau)$ or equivalently
  $\tbf{u}_{y,x}([0,\tau])\subseteq U$.  Therefore,
  $\tbf{u}_{x,y}$ is enabled until time $\tau$.  This means
  there exists a $y^\pr.y\xrightarrow{\tbf{u}_{y,x}}y^\pr$.

  From (\ref{eqn:rate-difference}),(\ref{eqn:supervisory})
  and that $||y-x||\leq \epsilon$, we get
  $(x^\pr-y^\pr)=(x-y)\exp((A+BC)\tau)\leq
  ||\epsilon\exp((A+BC)\tau)||$.  Since $(A+BC)$ has all
  eigenvalues with negative real part, so
  $||\epsilon\exp((A+BC)\tau)||$ tends exponentially to
  zero as $\tau\ra\infty$.  Therefore, we can choose
  sufficiently large $\tau$ such that
  $||x^\pr-y^\pr||<\eta/2$.
\end{proof}
\begin{res}\label{res:main}
  Let $\Sigma=\left<A_{n\times n},B_{n\times
      m},U,\bigcup_{t\in\mathbb{R}^+}U^{[0,t]}\right>$ be a
  locally asymptotically stabilizable linear control system with open and
  bounded input set $U$ and a stabilization matrix $C$.
  Then for all $\epsilon>0$ and $\eta:0\leq
  \eta<\epsilon/2$, we can choose $\tau>0$ such that
  $||\epsilon\exp((A+BC)\tau)||<\eta/2$ and consequently
  $T^\tau(\Sigma)$ is $||C||\epsilon$-trimmed
  $\epsilon$-approximately bisimilar to
  $T^{\tau,\eta}(\Sigma)$.
\end{res}
\begin{proof}
  On the basis of Corollary~\ref{cor:time-stamp}, we choose
  a $\tau>0$ such that $||\epsilon\exp((A+BC)\tau)||<\eta/2$
  and from that we derived, $\forall
  y,x:||y-x||\leq\epsilon$, if $\tbf{u}\in
  (U_{-||C||\epsilon})^{[0,\tau]}$ and
  $x\xrightarrow{\tbf{u}} x^\pr$, then $\exists
  y^\pr.y\xrightarrow{\tbf{u}_{y,x}}y^\pr$ and
  $||y^\pr-x^\pr||<\eta/2$.

  Choose a relation $R\subset \mathbb{R}^n\times
  [\mathbb{R}^n]_{\eta}$ as $(x,y)\in R$ if and only
  if $||x-y||\leq\epsilon$.  We shall prove that $R$ is the
  required $(||C||\epsilon+\delta)$-trimmed
  $\epsilon$-approximate bisimulation relation.  For this we
  have to prove, by the Definition of trimmed-input
  approximate bisimulation, both the following (i) $R$
  $\epsilon$-approximately simulates
  $T_{-||C||\epsilon}^\tau(\Sigma,X)$ by
  $T^{\tau,\eta}(\Sigma)$ by $R$. (ii) $R^{-1}$
  $\epsilon$-approximately simulates
  $T_{-||C||\epsilon}^{\tau,\eta}(\Sigma)$ by
  $T^{\tau}(\Sigma)$.

  The proof of (i) is as follows.  Let $(y,x)\in R$.  The
  transitions in
  $T_{-||C||\epsilon}^\tau(\Sigma,X)$ are driven by
  inputs in
  $(U^{[0,\tau]})_{-||C||\epsilon}=(U_{-||C||\epsilon})^{[0,\tau]}$
  (refer to definition of trimmed OIMTS and
  Proposition~\ref{prop:trimming}).  Let $(x,y)\in R$ and
  $x\xrightarrow{\tbf{u}} x^\pr$ be a transition in
  $T_{-||C||\epsilon}^\tau(\Sigma)$.  $(x,y)\in R$
  implies $||y-x||\leq\epsilon$ and hence by the choice of
  $\tau$ as stated in at the beginning of this proof, we
  have that there exists $y^\pr$ such that
  $y\xrightarrow{\tbf{u}_{y,x}} y^\pr$ and
  $||y^\pr-x^\pr||<\eta/2$.  Choose any
  $y^\dpr\in[\mathbb{R}^n]_{\eta}$ such that
  $||y^\pr-y^\dpr||\leq \eta/2$.  It is easy to see that
  such a point $y^\dpr$ exists on the grid
  $[\mathbb{R}^n]_{\eta}$.  Then by the way the
  state-quantized transition relation $\xrightarrow[\eta]{}$
  is defined, we get that
  $y\xrightarrow[\eta]{\tbf{u}_{y,x}} y^\dpr$ because
  $y\xrightarrow{\tbf{u}_{y,x}} y^\pr$,
  $||y^\pr-y^\dpr||\leq \eta/2$ and
  $y^\dpr\in[\mathbb{R}^n]_{\eta}$.  By triangular
  inequality, $||y^\dpr-x^\pr||\leq
  ||y^\pr-x^\pr||+||y^\pr-y^\dpr||\leq\eta/2+\eta/2=\eta\leq\epsilon$
  which means that $(x,y^\dpr)\in R$.  Also, by the way $R$
  was chosen earlier, it is an $\epsilon$ proximate
  relation.  This completes the proof of (i).

  We prove (ii) as follows.  Let $(y,x)\in R^{-1}$.  The
  transitions in $T^{\tau,\eta}_{-||C||\epsilon}$
  are driven by input trajectories in
  $(U^{[0,\tau]})_{-||C||\epsilon}=(U_{-||C||\epsilon})^{[0,\tau]}$
  by the definition of the trimmed OIMTS.  If
  for any
  $\tbf{u}\in(U_{-||C||\epsilon})^{[0,\tau]}$ we
  have $y\xrightarrow[\eta]{\tbf{u}}y^\pr$, then there
  exists $y^\dpr\in \mathbb{R}^n$ such that
  $||y^\dpr-y^\pr||<\eta/2$ and
  $y\xrightarrow{\tbf{u}}y^\dpr$, by the definition of the
  state-quantized transition relation
  $\xrightarrow[\eta]{}$.  Since
  $\tbf{u}\in(U_{-||C||\epsilon})^{[0,\tau]}$, so
  by the choice of $\tau$ at the beginning of this proof, we
  have an $x^\pr$ such that
  $x\xrightarrow{\tbf{u}_{x,y}}x^\pr$ and
  $||x^\pr-y^\dpr||<\eta/2$.  So, by triangular inequality
  we get that
  $||x^\pr-y^\pr||\leq||x^\pr-y^\dpr||+||y^\pr-y^\dpr||\leq
  \eta/2+\eta/2=\eta\leq\epsilon$.  Therefore,
  $(y^\pr,x^\pr)\in R^{-1}$.  Also, by the way $R$ was
  chosen earlier,
  $R^{-1}$ is an $\epsilon$ proximate relation.  This
  completes the proof of (ii).
\end{proof}
\section{The final symbolic model}\label{sec:final}
Let a quantized linear control system be
$\widetilde{\Sigma}=\left<A_{n\times n},B_{n\times
    m},\widetilde{U},\widetilde{\mc{U}}\right>$ where
$\widetilde{U}$ is a finite subset of an open set and
bounded set $U$ and $\widetilde{\mc{U}}$ is a finite set
containing piecewise constant input trajectories with
co-domain $\widetilde{U}$.  Consider that its \emph{analog
  approximation} is $\Sigma=\left<A_{n\times n},B_{n\times
    m},U,\bigcup_{t\in\mathbb{R}^+}U^{[0,t]}\right>$ which
is  {locally asymptotically stabilizable with open and bounded input
  space} $U$ and a stabilization matrix $C$.

Then for any desired precision $\epsilon>0$,
we can choose any state-quantization parameter
$\eta:0<\eta<\epsilon$, such that for any time quantization
$\tau>0$ satisfying $||\epsilon\exp((A+BC)\tau)||<\eta/2$,
we get that $T^\tau(\Sigma)$ is $||C||\epsilon$-trimmed
$\epsilon$-approximately bisimilar to
$T^{\tau,\eta}(\Sigma)$ by Result~\ref{res:main}.

With $\tau,\eta$ chosen as above for a given $\epsilon$, the
$||C||\epsilon$-trimmed $\tau,\eta$ state-time quantized
OIMTS $T^{\tau,\eta}_{-||C||\epsilon}(\Sigma)$ will be
employed in controller synthesis after restricting to
quantized inputs.  Note that the trimming of $||C||\epsilon$
is necessary for the symbolic model to be sound with respect
to $T^\tau(\Sigma)$.

\tbf{Soundness, proximity and near-completeness:} Recall the
interpretation of near completeness in
Section~\ref{sec:interpretation}.  The final symbolic model
$T^{\tau,\eta}_{-||C||\epsilon}(\Sigma)$ is sound,
$\epsilon$-proximate and $2{||C||\epsilon}$-near complete
with respect to $T^{\tau}(\Sigma)$ since
$T^{\tau,\eta}(\Sigma)$ is $||C||\epsilon$-trimmed
$\epsilon$ approximately bisimilar to $T^{\tau}(\Sigma)$.
Note that the symbolic model taken for controller synthesis
is $T^{\tau,\eta}_{-||C||\epsilon}(\Sigma)$ but not
$T^{\tau,\eta}(\Sigma)$ because the latter may not be sound
with respect to $T^\tau(\Sigma)$.

\tbf{Restricting the trimmed open input symbolic model to
  the quantized input set of actual control system:} The
open and $||C||\epsilon$-trimmed input $\eta,tau$ state-time
quantized transition system $T^{\tau,\eta}_{-||C||\epsilon}$
is
$T^{\tau,\eta}_{-||C||\epsilon}(\Sigma)=\left<[\mathbb{R}^n]_\eta,(U_{-\rho})^{[0,\tau]},\xrightarrow[\eta]{},\mathbb{R}^n,id\right>.$
Then the final symbolic model restricted to the actual
input-quantized control system $\widetilde{\Sigma}$ which
may be used in controller synthesis will be the transition
system
$\widetilde{T^{\tau,\eta}_{-||C||\epsilon}}(\widetilde{\Sigma})=\left<[\mathbb{R}^n]_\eta,(\widetilde{\mc{U}}\cap(U_{-\rho})^{[0,\tau]}),\xrightarrow[\eta]{},\mathbb{R}^n,id\right>.$

\tbf{Reducing the number of edges of the symbolic model:} If
$n$ is the number of representative points in the quantized
state space restricted to a desired compact set, then the
number of labeled edges emanating from any representative
point may far exceed $n$, because at each of the
representative state points, the number of input labels is
equal to the cardinality of
$\widetilde{\mc{U}}\cap\left(U_{-||C||\epsilon)}\right)^{[0,1]}$,
which could be very large.  Instead we may select, by
heuristic computations at each representative point, only a
subset of the input trajectories whose reach set, by the
transition relation $\xrightarrow[\eta]{}$, covers at least
all the reachable points in the finite state-quantized
space.  This would eliminate many labeled edges whose reach
points are the same as that of the former selected labeled
edges, while the \emph{sub-graph} so obtained is as complete
as
$\widetilde{T_{-||C||\epsilon}^{\tau,\eta}}(\widetilde{\Sigma})$.
Similar constructions have been discussed
in~\cite{pgt08,2012-majumdar-approximately}.

\section{Example}
\label{sec:example}
We take a quantized input linear system
$\widetilde{\Sigma}=\left<A_{n\times n},B_{n\times
    m},\widetilde{U},\widetilde{\mc{U}}\right>$ with $A=\left[\begin{array}{lcr}
    \mbox{0} & 1\\
    \mbox{-1} & 2\\
  \end{array}\right]$, $B=\left[\begin{array}{lcr}\mbox{0}\\\mbox{1} \end{array}
\right]$,
$\widetilde{U}=\{-0.49,-0.48,-0.47,...,-0.1,0,0.1,...,0.47,0.48,0.49\}$
and $\widetilde{\mc{U}}=\{\tbf{u}:\forall
t\in\mathbb{R}_{\geq 0}.\tbf{u}(t)\in\widetilde{U}\wedge\tbf{u}(t)=\tbf{u}\left(0.01*floor(t/0.01)\right)\}$
where $floor(.)$ denotes the greatest integer smaller than
the argument.  Then the analog input approximation of
$\widetilde{\Sigma}$ could be $\Sigma=\left<A_{n\times
    n},B_{n\times m},U,\mc{U}\right>$ with $U=]-0.5,0.5[$
which is an open set and $\mc{U}$ as the set of all
piecewise continuous input trajectories with co-domain $U$.

$A$ has both eigenvalues equal to $+1$ and so $\Sigma$ is
unstable.  But the system has a stabilization matrix
$C=[0~-4]$ because $(A+BC)=\left[\begin{array}{lcr}
    \mbox{0} & 1\\
    \mbox{-1} & 2\\
  \end{array}\right]+\left[\begin{array}{lcr}\mbox{0}\\\mbox{1} \end{array}
\right].[0~-4]=\left[\begin{array}{lcr}
    \mbox{0} & 1\\
    \mbox{-1} & -2\\
  \end{array}\right]$ has both eigenvalues equal to $-1$,
which is negative.

We are given a desired precision $\epsilon=0.12$.  We have
to determine the state-time quantization parameters
$\eta,\tau$ and the trimming parameter $\rho$.  We may take
$\eta$ to be anything less than $\epsilon=0.12$. Let
$\eta=0.1$.  Then the required amount of trimming is $\rho =
||C||\epsilon=4\times 0.14=0.48$ from Result~\ref{res:main}.
In fact, $\rho$ could be anything greater than or equal to
$0.48$ but should be at least $0.48$.  Note that the
derivation of $\rho$ is independent of $\tau$ which we have
not yet determined.  We take $\tau=1$ and demonstrate that
this particular choice of $\tau$ is valid.  For this we have
to prove that $\epsilon||\exp((A+BC)\tau)||<\eta/2$.  We get
$||\exp((A+BC)\tau)||\leq \exp(-1.\tau)$ because $(A+BC)$
has both eigenvalues equal to $-1$.  Then
$\epsilon||\exp((A+BC)\tau)||<0.12\exp(-1\times
1)=0.044<0.05=\eta/2$.  So, $\tau=1$ is a valid time
quantization parameter.

Consequently, the trimmed input trajectory set of the
symbolic model, taking $\rho=0.48$ and $\tau=1$, is
$(U_{-\rho})^{[0,\tau]}=\left((]-5,5[)_{-0.48}\right)^{[0,1]}=\left(]-4.52,4.52[\right)^{[0,1]}$.
Then the the state-time quantized symbolic abstraction of
$\Sigma$
is $$T^{1,0.1}_{-0.48}(\Sigma)=\left<[\mathbb{R}^2]_{0.1},\left(]-4.52,4.52[\right)^{[0,1]},\xrightarrow[0.1]{},\mathbb{R}^2,id\right>.$$
This symbolic model is sound and $0.12$-proximate as proved
in Section~\ref{sec:interpretation}.  Also,
$T^{1,0.1}_{-0.48}(\Sigma)$ is $2\times 0.48=0.96$-near
complete with respect to $T^{1}(\Sigma)$ in the sense that
$T^{1}_{0.96}(\Sigma)$ is approximately simulated by
$T^{1,0.1}_{-0.48}(\Sigma)$ as proved in
Section~\ref{sec:interpretation}.

Finally, a finite symbolic model can be obtained for any
compact region of state space and the actual quantized input
trajectory set $\widetilde{\mc{U}}$ by restricting
$T^{1,0.1}_{-0.48}(\Sigma)$ to the compact region and
$\widetilde{\mc{U}}$.  The restriction to
$\widetilde{\mc{U}}$ is defined in Section~\ref{sec:final}.
Furthermore, we select only a subset of the total number of
edges at each representative point to obtain a sub-graph
whose number of edges emanating from each representative
point is less than the total number of representative points
in the finite model, while the \emph{sub-graph} so obtained
is as complete as
$\widetilde{T_{-||C||\epsilon}^{\tau,\eta}}(\widetilde{\Sigma})$.  This is
explained in Section~\ref{sec:final}.  Since similar
constructions have been discussed
in~\cite{pgt08,2012-majumdar-approximately}, we do not
construct the actual model in our paper for this example.
However, we shall work out an illustration of the
relationship between the input at a representative point in
the quantized state space, and the corresponding supervisory
feedback at a point in the original state space which is
symbolically related to the representative point.  We shall
also obtain a quantized supervisory feedback which lies
within $\widetilde{\mc{U}}$, corresponding to the analog
supervisory feedback.

Two points are said to be symbolically related if the norm
of the difference between them is less than $\epsilon$.  $z=
(0.23, -0.24)$ is symbolically related to $x= (0.2, -0.2)$
because
$||z-x||=\max\{0.23-0.2,0.24-0.2\}=0.04<\epsilon=0.12$.
Notice that $x$ is a point in the quantized state space with
$\eta=0.1$.  We give a constant input $\tbf{u}(t)=1.1$
$\forall t\in[0,1]$ driving from $x=(0.2,-0.2)$.  Notice
that
$\tbf{u}\in\widetilde{\mc{U}}\cap\left(]-4.52,4.52[\right)^{[0,1]}$.
By computation, we get that
$\tbf{x}(1,(0.2,-0.2),\tbf{u})=(0.56,1.36)$.  By the
definition of the transition relation
$\xrightarrow[0.1]{\tbf{u}}$ defined on the state-time
quantized symbolic model, we get that
$x=(0.2,-0.2)\xrightarrow[0.1]{\tbf{u}}(0.6,1.4)=x^\pr$
where $(0.56,1.36)$ is rounded off to $(0.6,1.4)$.
\begin{figure*}[h]
\vspace{1em}
  \centering \setlength\fboxsep{0pt}
  \setlength\fboxrule{0.5pt}
  {\includegraphics[scale=0.3, bb=300 0 900
    300]{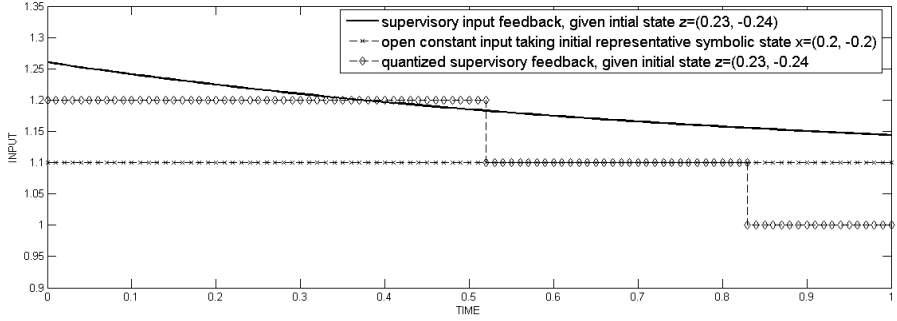}}
  \caption{Corresponding analog and quantized supervisory
    feedback inputs at $z=(0.23,-0.24)$ for constant input
    $u=1.1$ at symbolic point
    $x=(0.2,-0.2)$.}\label{fig:input}
\end{figure*}
Then the corresponding analog supervisory input driving from
$z$ is given as $\tbf{u}_{z,x}=1.1+C(\tbf{z}(t)-\tbf{x}(t))$
according to Equation~\ref{eqn:supervisory}.  The graph of
the analog supervisory input is displayed in
Figure~\ref{fig:input}.  By inducing $\tbf{u}_{z,x}$ at $z$,
we reach the point $z^\pr=(0.56,1.35)$ at time $t=1$.  Since
$||z^\pr-x^\pr||=\max\{0.6-0.56,1.4-1.35\}=0.05<\epsilon=0.12$,
therefore $z^\pr$ is symbolically related to $x^\pr$.  Thus
our assertion that the symbolic model is $0.12$-proximate is
validated in this specific example.

Next, we heuristically quantized the supervisory input
$\tbf{u}_{z,x}$ to lie in the quantized input trajectory
space $\widetilde{\mc{U}}$, and obtained the quantized
supervisory feedback input displayed in
Figure~\ref{fig:input}.  There are also formal approaches to
feedback
quantization~\cite{2000-brockett-quantized,1990-delchamps-stabilizing,2005-fu-sector}.
Driving with the quantized input starting from $z=(0.23,
-0.24)$, we reached the point $z^\dpr=(0.58,1.38)$ at 1
second.  Since
$||z^\dpr-x^\pr||=\max\{0.6-0.58,1.4-1.38\}=0.02<\epsilon=0.12$,
therefore $z^\dpr$ is also symbolically related to $x^\pr$.
This means that, for this specific illustration, the
proximity of $0.12$ is valid even after quantizing the
supervisory input.

From Figure~\ref{fig:input}, the difference between
$\tbf{u}_{z,x}$ and $\tbf{u}$ is less that
$0.48=||C||\epsilon=\rho$. Therefore, our assertion in
Theorem~\ref{thm:enabled} is valid for this example.  Also
the difference between the \emph{quantized} supervisory
feedback input and $\tbf{u}$ is less than $0.48$.
\section{Conclusion}
While allowing supervisory feedback to relate inputs between
two transition systems, we have found a formal way of
parametrization of completeness of a state-time quantized
symbolic model with respect to the time quantized system
model.  We demonstrated how sound state-time quantized
symbolic models of \emph{possibly unstable} but
stabilizable, bounded input and already input-quantized
linear systems can be built with arbitrarily small proximity
and trimming (near-completeness), with respect to the
time-quantized model.  In future, we would like to extend
this work to construct sound, near-complete, and proximate
symbolic models for non-linear systems.
\section*{Appendix}
\begin{defn}
\label{defn:supervisory}
  Let $U\subseteq\mathbb{R}^m$ for some
  $m\in\mathbb{N}$. Then for $n\in\mathbb{N}$, a function
  $k:\mathbb{R}^n\times\mathbb{R}^n\times U$ is called a
  supervisory function iff all the following
  hold.  
  \begin{enumerate}
\item $k$ is continuously differentiable on
  $\mathbb{R}^{2n}\backslash\Delta$ where
  $\Delta=\{(x,x):x\in\mathbb{R}^n\}$;
\item $k(y,x,u)=u$ $\forall$ $(y,x)\in\Delta$;
\end{enumerate}
\end{defn}
We say that a supervisory function is linear if it is of the
form $k(y,x,u)=u+C(y-x)$ for some matrix $C$.  It is easy to
see that $\left(u+C(y-x)\right)$ moves out of any bounded
input set because the supervisory function translates the
original input set by an amount $C(y-x)$.  Therefore,
there can not be a linear supervisory function on a bounded
input set.

We now state the general global asymptotic stabilizability
assumption and prove that everywhere divergent linear
systems with  {bounded input set} can not be globally
asymptotically stabilized by any kind of supervisory
function.

\tbf{Notation}: A function $\beta:\mathbb{R}_{\geq
  0}\times\mathbb{R}_{\geq 0}\ra \mathbb{R}_{\geq 0}$ is
called a $\mathcal{KL}_\infty$ function if $\beta(r,.)$ is
increasing function in $r$ such that $\beta(0,.)=0$; and
$\beta(.,t)$ asymptotically tends to zero as $t\ra\infty$.

The following stabilizability assumption, which we call the
global asymptotic stabilizability assumption, was discussed
in~\cite{2008-tabuada-approximate}.
\begin{defn}\label{defn:sa1} A control system
  with input set $U$ and state space $\mathbb{R}^n$ is
  said to be globally asymptotically stabilizable if there
  exists a supervisory function
  $k:\mathbb{R}^n\times\mathbb{R}^n\times U\ra U$ enforcing
  the following estimate for all $x,y\in\mathbb{R}^n$,
  $\tbf{u}\in\mathcal{U}$ and $t\in\mathbb{R}_0^{+}$
\begin{equation}\label{eqn:asymptotic}
||\tbf{x}(t,x,\tbf{u})-\tbf{y}(t,y,k(\tbf{y},\tbf{x},\tbf{u}))||\leq\beta(||x-y||,t).
\end{equation}
where $\beta$ is a $\mc{KL}_\infty$ function.
\end{defn}

A linear system $\Sigma=\left<A_{n\times n},B_{n\times
    m},U,\mathcal{U}\right>$ is everywhere divergent if all
the eigenvalues of $A$ have positive real part.  Consider
$\lambda_{min}$ as the eigenvalue of $A$ with minimum real
part.  Consider $\Sigma$ as everywhere divergent and so
$Re(\lambda_{min})$ is positive.  We proceed to demonstrate
that for the everywhere divergent linear system $\Sigma$, if
the input set $U$ is bounded such that $||u||<M~\forall
u\in U$, then any two state trajectories starting at a
distance greater than $\frac{4||B||M}{Re(\lambda_{min})}$
can never come arbitrarily close, irrespective of what pair
of input trajectories drives the two state trajectories.
This would mean that the system can not be globally
asymptotically stabilized.  The proof is as follows.

We know for a linear system 
\begin{equation}\label{eqn:linear}
\begin{split}
 \tbf{x}(x,\tau,\tbf{u})-\tbf{y}(y,\tau,\tbf{v})=\exp(A\tau)(x-y)+\int_0^\tau\exp(A(\tau-t))B(\tbf{u}-\tbf{v})(t)dt
\end{split}
\end{equation}

By using reverse triangular inequality on (\ref{eqn:linear})
we get 
\begin{equation}\label{eqn:triangular}
\begin{split}
&\left|\left|\tbf{x}(x,\tau,\tbf{u})-\tbf{y}(y,\tau,\tbf{v})\right|\right|\\
&\geq \lel|\exp(A\tau)(x-y)\rer|-\lel|\int_0^\tau\exp(A(\tau-t)B(\tbf{u}-\tbf{v})(t)\rer|
\end{split}
\end{equation}

For all possible pairs of trajectories $\tbf{u}$ and
$\tbf{v}$, we have that $\lel|\tbf{u}-\tbf{v}\rer|\leq 2M$
since the norm of inputs is bounded by $M$.

Then choose $x,y$ such that $||x-y||>\frac{4M||B||}{Re(\lambda_{min})}$.
Putting these bounds in (\ref{eqn:triangular}) we get
\begin{equation}\label{eqn:bound}
\begin{split}
&\left|\left|\tbf{x}(x,\tau,\tbf{u})-\tbf{y}(y,\tau,\tbf{v})\right|\right|\\
&\geq 2M||B||\lel|\exp(A\tau)\left(\frac{2}{Re(\lambda_{min})}-\int_0^\tau\exp(-At)dt\right)\rer|
\end{split}
\end{equation}

Again using reverse triangular inequality we get 
\begin{equation}\label{eqn:bound1}
\begin{split}
&\left|\left|\tbf{x}(x,\tau,\tbf{u})-\tbf{y}(y,\tau,\tbf{v})\right|\right|\\
&\geq 2M||B||||\exp(A\tau)||\left(\frac{2}{Re(\lambda_{min})}-\lel|\int_0^\tau\exp(-At)dt\rer|\right)
\end{split}
\end{equation}

Since $\lel|\int_0^\tau\exp(-At)dt\rer|\leq
\int_0^\tau||\exp(-At)||dt$ So
$$\left(\frac{2}{Re(\lambda_{min})}-\lel|\int_0^\tau\exp(-At)dt\right|\rer)\geq
\left(\frac{2}{Re(\lambda_{min})}-\int_0^\tau||\exp(-At)||dt\right).$$  

Furthermore, we have that $||\exp(-At)||\leq
\exp(-Re(\lambda_{min})t)$ since $\lambda_{min}$ is the
eigenvalue with minimum real part.  Substituting we get
$$\left(\frac{2}{Re(\lambda_{min})}-\lel|\int_0^\tau\exp(-At)dt\right|\rer)\geq\left(\frac{2}{Re(\lambda_{min})}-\int_0^\tau\exp(-Re(\lambda_{min})t)dt\right)$$
$$\geq
\left(\frac{2}{Re(\lambda_{min})}-\frac{1}{Re(\lambda_{min})}\right)~~~\text{since}
~\int_0^\infty\exp(-Re(\lambda_{min})t)dt=
1/Re(\lambda_{min}).$$

Substituting in (\ref{eqn:bound1}) we get
\begin{equation}
\begin{split}
  &\left|\left|\tbf{x}(x,t,\tbf{u})-\tbf{y}(y,t,\tbf{v})\right|\right|
  \geq
  \frac{2M||B||||\exp(A\tau)||}{Re(\lambda_{min})}\\
    \end{split}
\end{equation}

$\frac{2M||B||||\exp(A\tau)||}{Re(\lambda_{min})}$ keeps
increasing in $\tau$ because $A$ has all eigenvalues with
positive real part.  

This means that, the linear system being \emph{everywhere
  divergent} with {norm of inputs upper bounded by} $M>0$,
if the starting points $x$ and $y$ are such that
$||x-y||>\frac{4M||B||}{Re(\lambda_{min})}$ where
${Re(\lambda_{min})}$ is the eigenvalue with minimum real
part (which is positive), then $\tbf{x}$ and $\tbf{y}$ can
not come arbitrarily close for any possible input
trajectories $\tbf{u}$ and $\tbf{v}$ inside the bounded
input set driving $\tbf{x}$ and $\tbf{y}$ respectively.
This completes the proof.

Hence an everywhere divergent linear system whose {input set
  is bounded} can not be globally asymptotically stabilized.
But such a system may still be locally asymptotically
stabilizable and an example was shown in
Section~\ref{sec:finite-abstraction}.  However if the input
space has no boundaries, then global asymptotic
stabilizability of linear systems is equivalent to local
asymptotic stabilizability.

\bibliographystyle{plain}


\end{document}